\documentclass[10pt,a4paper,reqno]{amsart}

\usepackage[english]{babel}

\usepackage{color} \definecolor{bleu_sombre}{rgb}{0,0,0.6}  \definecolor{rouge_sombre}{rgb}{0.8,0,0}\definecolor{vert_sombre}{rgb}{0,0.6,0}
\usepackage[plainpages=false,colorlinks,linkcolor=bleu_sombre,citecolor=rouge_sombre,urlcolor=vert_sombre,breaklinks]{hyperref}

\usepackage{amsmath,amssymb,amsthm,graphicx,amsfonts,url,color,enumerate,dsfont,stmaryrd,mathabx,mathrsfs}

\theoremstyle{plain}
\newtheorem{theorem}{{Theorem}}[section] 
\newtheorem*{theorem*}{{Theorem}}
\newtheorem{proposition}[theorem]{Proposition}
\newtheorem*{proposition*}{Proposition}

\newtheorem*{corollary*}{Corollary}
\newtheorem{lemma}[theorem]{Lemma}
\newtheorem*{lemma*}{Lemma}
\theoremstyle{definition}

\newtheorem*{definition*}{Definition}
\theoremstyle{remark}
\newtheorem{remark}[theorem]{Remark}

\makeatletter

\@addtoreset{equation}{section}  
\makeatother

\renewcommand{\leq}{\leqslant}	\renewcommand{\geq}{\geqslant}

\newcommand{\R}{\mathbb{R}}
\newcommand{\C}{\mathbb{C}}
\newcommand{\N}{\mathbb{N}}	

\newcommand{\Dom}{\mathsf{Dom}}
\newcommand{\Sp}{\mathsf{Sp}}
\newcommand{\dx}{\mathrm{d}}
\newcommand{\eps}{\varepsilon}
\newcommand{\eff}{\mathsf{eff}}
\renewcommand{\Im}{\mathrm{Im}}
\renewcommand{\Re}{\mathrm{Re}}

\newcommand{\fonc}[4] { \left\{ \begin{array}{ccc} #1 & \to & #2 \\ #3 & \mapsto & #4 \end{array} \right. }

\renewcommand{\a}{\alpha}\newcommand{\g}{\gamma}\newcommand{\e}{\varepsilon}\newcommand{\z}{\zeta} \newcommand{\Th}{\Theta}\renewcommand{\l}{\lambda}\newcommand{\m}{\mu}\newcommand{\f}{\varphi}\newcommand{\p}{\psi}\renewcommand{\O}{\Omega}
\newcommand{\Hc}{{\mathcal H}}\newcommand{\Kc}{{\mathcal K}}\newcommand{\Lc}{{\mathcal L}}\newcommand{\Qc}{{\mathcal Q}}\newcommand{\Zc}{{\mathcal Z}}
\newcommand{\abs}[1]{\left\vert #1\right\vert}\newcommand{\nr}[1]{\left\Vert #1\right\Vert}\newcommand{\innp}[2]{\left< #1 , #2 \right>}

\newcommand{\LL}{\mathscr{L}}
\newcommand{\LLo}{\mathscr{L}}
\newcommand{\LLeff}{{\mathscr{L}}_\eff}\newcommand{\LLeffh}{{\mathscr{L}}_{h,\eff}}
\newcommand{\LLbot}{{\mathscr{L}}^\bot}
\newcommand{\hatLL}{\widehat{\mathscr{L}}}
\newcommand{\LLp}{\mathscr{L}}

\newcommand{\LLh}{\mathscr{L}_h}
\newcommand{\LLDir}{\mathscr{L}^{\mathsf{Dir}}_{\varepsilon}}

\newcommand{\LLDireff}{\mathscr{L}^{\mathsf{Dir}}_{\varepsilon,\eff}}

\newcommand{\hatLLDR}{\widehat{\mathscr{L}}^{\mathsf{DR}}_{\alpha}}
\newcommand{\LLDReff}{\mathscr{L}^{\mathsf{DR}}_{\alpha,\eff}}
\newcommand{\LLe}{\mathscr{L}_{\varepsilon}}
\newcommand{\hatLLe}{\widehat{\mathscr{L}}_{\varepsilon}}

\newcommand{\PP}{\mathscr{P}}

\newcommand{\Qo}{Q}
\newcommand{\hatQ}{\widehat Q}

\newcommand{\Qeff}{Q_\eff}
\newcommand{\Qbot}{Q_\bot}

\newcommand{\Qe}{Q_\e}  \newcommand{\hatQe}{\widehat Q_\e} 

\renewcommand{\AA}{\mathcal A}
\newcommand{\hAA}{\widehat {\mathcal A}}

\usepackage[normalem]{ulem}
\definecolor{DarkGreen}{rgb}{0,0.5,0.1} 

\newcommand\soutD{\bgroup\markoverwith
{\textcolor{DarkGreen}{\rule[.5ex]{2pt}{1pt}}}\ULon}
\newcommand{\Hm}[1]{\leavevmode{\marginpar{\tiny%
$\hbox to 0mm{\hspace*{-0.5mm}$\leftarrow$\hss}%
\vcenter{\vrule depth 0.1mm height 0.1mm width \the\marginparwidth}%
\hbox to
0mm{\hss$\rightarrow$\hspace*{-0.5mm}}$\\\relax\raggedright #1}}}

\newcommand\soutP{\bgroup\markoverwith
{\textcolor{blue}{\rule[.5ex]{2pt}{1pt}}}\ULon}

\begin{document}
\title[Reduction of dimension]{Reduction of dimension as a consequence of norm-resolvent
convergence and applications}

\author{D. Krej{\v{c}}i{\v{r}}{\'{\i}}k}
\address[D. Krej{\v{c}}i{\v{r}}{\'{\i}}k]{
Department of Mathematics, Faculty of Nuclear Sciences and Physical Engineering, Czech Technical University in Prague,
Trojanova 13, 12000 Prague 2, Czech Republic}
\email{david.krejcirik@fjfi.cvut.cz}

\author{N. Raymond}
\address[N. Raymond]{IRMAR, Universit\'e de Rennes 1, Campus de Beaulieu, F-35042 Rennes cedex, France}
\email{nicolas.raymond@univ-rennes1.fr}

\author{J. Royer}
\address[J. Royer]{Institut de math\'ematiques de Toulouse, Universit\'e Toulouse 3, 118 route de Narbonne, F-31062 Toulouse cedex 9, France}
\email{julien.royer@math.univ-toulouse.fr}

\author{P. Siegl}
\address[P. Siegl]{
Mathematical Institute, 
University of Bern,
Alpeneggstrasse 22,
3012 Bern, Switzerland
\& On leave from Nuclear Physics Institute ASCR, 25068 \v Re\v z, Czech Republic}
\email{petr.siegl@math.unibe.ch}

\begin{abstract}
This paper is devoted to dimensional reductions via the norm resolvent convergence. We derive explicit bounds on the resolvent difference as well as spectral asymptotics. The efficiency of our abstract tool is demonstrated by its application on seemingly different PDE problems from various areas of mathematical physics; all are analysed in a unified manner now, known results are recovered and new ones established.
\end{abstract}

\subjclass[2010]{Primary: 81Q15; Secondary: 35P05, 35P20, 58J50, 81Q10, 81Q20}  

\keywords{reduction of dimension, norm resolvent convergence,
Born-Oppenheimer approximation, thin layers, quantum waveguides, effective Hamiltonian}

\maketitle

\newcommand{\Cor}[1]{{\color{red}#1}}

\section{Introduction}

\subsection{Motivation and context}

In this paper we develop an abstract tool for dimensional reductions via the norm resolvent convergence obtained from variational estimates. The results are relevant in particular for PDE problems, typically Schr\"odinger-type operators depending on an asymptotic parameter having various interpretations (semiclassical limit, shrinking limits, large coupling limit, etc.). In applications, our resolvent estimates lead to accurate spectral asymptotic results for eigenvalues lying in a suitable region of the complex plane. Moreover, avoiding the traditional min-max approach, with its fundamental limitations to self-adjoint cases, we obtain an effective operator, the spectrum of which determines the spectral asymptotics. The flexibility of the latter is illustrated on a non-self-adjoint example in the second part of the paper.

The power of our approach is demonstrated by a unified treatment of diverse classical as well as latest problems occurring in mathematical physics such as:
\begin{enumerate}[-]
	\item semiclassical Born-Oppenheimer approximation,
	\item shrinking tubular neighborhoods of hypersurfaces subject to various boundary conditions, 
	\item domains with very attractive Robin boundary conditions.
\end{enumerate} 
In spite of the variety of operators, asymptotic regimes, and techniques considered in the previous literature, all these results are covered in our general abstract and not only asymptotic setting. Our first result (Theorem \ref{theo.main}) gives a norm resolvent convergence towards a tensorial operator in a general self-adjoint setting. A remarkable feature is that only two quantities need to be controlled:
the size of a commutator of a ``longitudinal operator'' with spectral projection on low lying ``transverse modes'' and the size of the ``spectral gap'' of a ``transverse operator'', 
see~\eqref{a.def} and~\eqref{gamma}, respectively.
Although the latter is also very natural it was hardly visible in existing literature due to many seemingly different technical steps as well as various ways how these quantities enter. As particular cases of the application of Theorem \ref{theo.main}, we recover, in a short manner, known results for quantum waveguides 
(see for instance \cite{Duclos95}, \cite{KS12}, \cite{KR14} or \cite{
KRT15}) and cast a new light on Born-Oppenheimer type results (see \cite{LT13}, \cite{WT14}, \cite{J14} or \cite[Sec.~6.2]{Ray17}). To keep the presentation short we deliberately do not strive for the weakest possible assumptions in examples, although the abstract setting allows for many further generalizations and it clearly indicates how to proceed.

In the second part of the paper, we prove, in the same spirit as previous results, the norm convergence result for a non-self-adjoint Robin Laplacian, see Theorem \ref{th.n-s-a-Robin}. It will partially generalize previous works in the self-adjoint (see \cite{PP-eh}, \cite{KKR16} and \cite{Oliveira}) and in the non-self-adjoint (see \cite{BorisovKr12}) cases. 

As a matter of fact, the crucial step in all the proofs of the paper is an abstract lemma (see Lemma~\ref{lem-diff-resolv}) of an independent interest. It provides a norm resolvent estimate from variational estimates, which is particularly suitable for the analysis of operators defined via sesquilinear forms.

\subsection{Reduction of dimension in an abstract setting
and self-adjoint applications} \label{sec-intro-abstract}\label{sec.general-reduction}

We first describe the reduction of dimension for an operator of the form 
\begin{equation} \label{def-LL}
\LL = S^* S + T, \qquad T = \bigoplus _{s \in \Sigma} T_s, 
\end{equation}
acting on the Hilbert space $\Hc = \bigoplus_{s \in \Sigma} \Hc_s$. 
The norm and inner product in~$\Hc$ will be denoted
by~$\|\cdot\|$ and~$\langle\cdot,\cdot\rangle$, respectively; 
the latter is assumed to be linear in the second argument.

Here $\Sigma$ is a measure space and $T_s$ is a self-adjoint non-negative operator
on a Hilbert space $\Hc_s$ for all $s \in \Sigma$. Precise definitions will be given in Section~\ref{sec-abstract}. A typical example is 
the Schr\"odinger operator
\[
H = (-i\hbar\partial_s)^2 + (-i\partial_t)^2+V(s,t)\,,
\]
acting on $L^2(\R_s \times \R_t)$. We consider a function $s \mapsto \gamma_s$ such that 
\begin{equation}\label{gamma}
\gamma = \inf_{s \in \Sigma} \g_s > 0\,.
\end{equation}
Then we denote by $\Pi_s \in \Lc(\Hc_s)$ the spectral projection of $T_s$ on $[0,\gamma_s)$, and we set $\Pi_s^\bot = \mathrm{Id}_{\Hc_s} - \Pi_s$. We denote by $\Pi$ the bounded operator on $\Hc$ such that for $\Phi \in \mathcal{H}$ and $s \in \Sigma$ we have $(\Pi \Phi)_s = \Pi_s \Phi_s$. We similarly define $\Pi^\bot \in \mathcal{L}(\mathcal{H})$. 
Our purpose is to compare some spectral properties of the operator $\LL$ with those of the simpler operator 
\begin{equation} \label{def-LLeff}
\LLeff = \Pi \LL \Pi.
\end{equation}
This is an operator on $\Pi \Hc$ with domain $\Pi \Hc \cap \Dom(\LL)$. 

In fact, we will first compare $\LL$ with 
\begin{equation} \label{def-hatLL}
\hatLL = \Pi \LL \Pi + \Pi^\bot \LL \Pi^\bot.
\end{equation}
Then $\LLeff$ and $\LLbot$ will be defined as the restrictions of $\hatLL$ to $\Pi\Hc$ and $\Pi^\bot \Hc$, respectively, so that
\[
\hatLL = \LLeff \oplus \LLbot.
\]
We will give a sufficient condition for $z \in \rho(\hatLL)$ to be in $\rho(\LL)$ and, in this case, an estimate for the difference of the resolvents. Then, since $\Pi \Hc$ and $\Pi^\bot \Hc$ reduce $\hatLL$, it is not difficult to check that far from the spectrum of $\LLbot$ the spectral properties of $\hatLL$ are the same as those of $\LLeff$, so we can state a similar statement with $\hatLL$ replaced by $\LLeff$. In applications, we can for instance prove that the first eigenvalues of $\LL$ are close to the eigenvalues of the simpler operator $\LLeff$.
 
We assume that $\Dom(S)$ is invariant under $\Pi$, that $[S,\Pi]$ extends to a bounded operator on~$\Hc$, and we set
\begin{equation}\label{a.def}
a = \frac{\|[S,\Pi]\|_{\mathcal{L}(\mathcal{H})}}{\sqrt{\gamma}}\,.
\end{equation}
For $z \in \C$, we also define 
\begin{equation} \label{def-eta}
\begin{aligned}
\eta_{1}(z)&= \frac{3}{\sqrt{2}} a^2 \g+\frac{6a}{\sqrt{2}}(1+a)|z|+\frac{3a}{\gamma\sqrt{2}}\left(2+\frac{a}{\sqrt{2}}\right)|z|^2\,,\\
\eta_{2}(z)&=\frac{3a}{\sqrt{2}}(1+a)+\frac{3a}{\gamma\sqrt{2}}\left(2+\frac{a}{\sqrt{2}}\right)|z|\,,\\
\eta_{3}(z)&=\frac{3a}{\sqrt{2}}\left(1+\frac{a}{\sqrt{2}}\right)+\frac{3a}{\gamma\sqrt{2}}\left(2+\frac{a}{\sqrt{2}}\right)|z|\,,\\
\eta_{4}(z)&=\frac{3a}{\gamma\sqrt{2}}\left(2+\frac{a}{\sqrt{2}}\right)\,.
\end{aligned}
\end{equation}

\begin{theorem}\label{theo.main}
Let $z\in \rho(\hatLL)$. If 
\[
1 -\eta_{1}(z) \|(\hatLL-z)^{-1}\|-\eta_{2}(z)>0,
\]
then $z\in\rho(\LLo)$ and 
\begin{multline*}
\|(\LLo-z)^{-1}-(\hatLL-z)^{-1}\|\\
\leq \eta_1(z) \|(\LLo-z)^{-1} \| \|(\hatLL-z)^{-1}\| + \eta_2(z) \|(\LLo-z)^{-1} \|+\eta_{3}(z)\|(\hatLL-z)^{-1}\|+\eta_{4}(z).
\end{multline*}
In particular,
\[
\|(\LLo-z)^{-1}\|\leq  \frac{ ( \eta_{3}(z)+1) \| (\hatLL-z)^{-1} \|+\eta_{4}(z)}{1- \eta_1(z) \|(\hatLL-z)^{-1}\|-\eta_2(z)}.
\]
\end{theorem}

In order to compare the resolvent of $\LL$ to the resolvent of $\LLeff$, this theorem is completed by the following easy estimate:

\begin{proposition}\label{prop.compl-theo.main}
We have $\Sp(\hatLL) = \Sp(\LLeff) \cup \Sp(\LLbot)$ and, for $z \in \rho(\hatLL)$ such that $z\notin [\gamma,+\infty)$,
\[
\nr{(\hatLL-z)^{-1} - (\LLeff-z)^{-1} \Pi}
\leq \frac 1 {\mathsf{dist}(z, [\gamma,+\infty))}\,.
\]
In this estimate, it is implicit that $(\LLeff-z)^{-1}$ is composed on the left by the inclusion $\Pi \Hc \to \Hc$.
\end{proposition}
\begin{remark}\label{rem.0}
These results cover a wide range of situations. In Section \ref{sec-examples}, we will discuss three paradigmatic applications. The space $\Sigma$ will be $\R$ or a submanifold of $\R^d$, $d \geq 2$. The set $\Hc_s$ is fixed, but the Hilbert structure thereon may depend on $s$. In our examples $(T_s)_{s\in\Sigma}$ is related to an analytic family of self-adjoint operators which are not necessarily non-negative. Nevertheless, under suitable assumptions, we can reduce ourselves to the non-negative case. Indeed, in our applications, for all $s\in\Sigma$, $T_{s}$ is bounded from below, independently of $s\in\Sigma$. Moreover, the bottom of the spectrum of $T_s$ will be an isolated simple eigenvalue $\mu_1(s)$. Then, we notice that $\inf_{s\in\Sigma}\mu_{1}(s)$ is well-defined and that $T_{s}-\inf_{s\in\Sigma}\mu_{1}(s)$ is non-negative. We denote by $u_1(s)$ a corresponding eigenfunction. We can assume that $\nr{u_1(s)}_\Hc = 1$ for all $s \in \Sigma$ and that $u_1$ is a smooth function of $s$. $\Pi_s$ is the projection on $u_{1}(s)$ and $\Pi \Hc$ can be identified with $L^2(\Sigma)$ via the map $\f \mapsto (s \mapsto \f(s) u_1(s))$. In particular $\LLeff$ can be seen as an operator on $L^2(\Sigma)$, which is what is meant by the ``reduction of dimension''. Finally, $\g_s$ is defined as the bottom $\mu_2(s)-\inf_{s\in\Sigma}\mu_{1}(s)$ of the remaining part of the spectrum and
\begin{equation} \label{def-gamma}
\g = \inf_{s} \mu_2(s)-\inf_{s}\mu_{1}(s) \leq \inf \Sp((\LL-\inf_{s\in\Sigma}\mu_{1}(s))^\perp)\,.
\end{equation}
We recall that we assume the spectral gap condition $\gamma>0$, see \eqref{gamma}.
\end{remark}

\subsection{The Robin Laplacian in a shrinking layer
as a non-self-adjoint application} \label{sec-intro-nsa-robin}

We now consider a reduction of dimension result in a non-self-adjoint setting, namely the Robin Laplacian in a shrinking layer.
Let $d \geq 2$. Here, $\Sigma$ is an orientable smooth (compact or non-compact)
hypersurface in $\R^d$ without boundary. 
The orientation can be specified by a globally defined unit normal vector field $n : \Sigma \to \mathbb S^{d-1}$. Moreover $\Sigma$ is endowed with the Riemannian structure inherited from the Euclidean structure defined on $\R^d$. 
We assume that $\Sigma$ admits a tubular neighborhood, i.e. for $\eps > 0$ small enough the map
\begin{equation} \label{def-Theta}
\Theta_\eps :  (s,t) \mapsto  {s + \eps t n(s)}
\end{equation}
is injective on $\overline{\Sigma} \times [-1,1]$ and defines a diffeomorphism from $\Sigma \times (-1,1)$ to its image. We set
\begin{equation} \label{def-Omega-eps}
\Omega = \Sigma \times (-1,1) \quad \text{and} \quad \Omega_{\eps} = \Theta_\eps (\Omega)\,.
\end{equation}
Then $\Omega_{\eps}$ has the geometrical meaning of a non-self-intersecting layer delimited by the hypersurfaces
\[
\Sigma_{\pm,\eps} = \Theta_\eps (\Sigma \times \{\pm 1\})\,.
\]
Moreover $\Sigma_{\pm,\eps}$ can be identified with $\Sigma$ via the diffeomorphisms
\[
\Theta_{\pm,\eps} : \fonc {\Sigma} {\Sigma_{\pm,\eps}} {s} {s \pm \eps n(s)\,.}
\]

Let $\alpha : \Sigma \to \C$ be a smooth bounded function. We set $\alpha_{\pm,\eps} = \alpha \circ \Theta_{\pm,\eps} ^{-1} : \Sigma_{\pm,\eps} \to \C$ and we consider on $L^2(\Omega_{\eps})$ the closed operator $\PP_{\eps,\alpha}$ (or simply $\PP_{\eps}$ if no risk of confusion) defined as the usual Laplace operator on $\Omega_{\eps}$ subject to the Robin boundary condition 
\begin{equation} \label{BC-Robin}
\frac {\partial u}{\partial n} + \alpha_{\pm,\eps} u = 0, \quad \text{on } \Sigma_{\pm,\eps}\,.
\end{equation}
\begin{remark}
Note that a very special choice of Robin boundary conditions 
is considered in this section. Indeed, the boundary-coupling functions 
considered on~$\Sigma_{+,\eps}$ and~$\Sigma_{-,\eps}$ 
are the same except for a switch of sign, see~\eqref{Robin.form}.
More specifically, $\alpha_{\pm,\eps}(s) = \alpha(s)$ for every $s \in \Sigma$
and~$n$ is an outward normal to~$\Omega_\eps$ on one of 
the connected parts~$\Sigma_{\pm,\eps}$ of the boundary~$\partial\Omega_\eps$,
while it is inward pointing on the other boundary.
This special choice is motivated 
by Parity-Time-symmetric waveguides \cite{BK,BorisovKr12}
as well as by a self-adjoint analogue considered in~\cite{Oliveira}. 
It is straightforward to extend the present procedure 
to the general situation of two different 
boundary-coupling functions on~$\Sigma_{+,\eps}$ and~$\Sigma_{-,\eps}$,
but then the effective operator will be $\eps$-dependent 
(in analogy with the Dirichlet boundary conditions,
see Proposition~\ref{cor.waveguides})
or a renormalization would be needed (cf.~\cite{KS12}).
\end{remark}

Our purpose is to prove that, at the limit when $\eps$ goes to 0, the operator $\PP_\eps$ converges in a norm-resolvent sense to a Schr\"odinger operator 
\[
\LLeff = -\Delta_\Sigma + V_\eff\,,
\]
on $\Sigma$. Here $-\Delta_\Sigma$ is the Laplace-Beltrami operator on $\Sigma$, and the potential $V_\eff$ depends both on the geometry of $\Sigma$ and on the boundary condition. More precisely we have
\begin{equation}\label{Veff}
V_\eff = |\alpha|^2 - 2\alpha \Re(\alpha) - \alpha (\kappa_1 + \dots + \kappa_{d-1}).
\end{equation}
Note that the sum of the principal curvatures is proportional to the mean curvature of $\Sigma$. Notice also that $H_\eff$ defines an (unbounded) operator on the Hilbert space $L^2(\Sigma)$. In particular $\PP_\eps$ and $H_\eff$ do not act on the same space.

We denote by $\Pi \in \mathcal{L}(L^2(\Omega))$ the projection on functions which do not depend on $t$: for $u \in L^2(\Omega)$ and $(s,t) \in \Omega$ we set 
\[
(\Pi u) (s,t) = \frac 1 2 \int_{-1}^1 u(s,\theta)\dx \theta\,. 
\]
Then we define $\Pi^\bot = \mathrm{Id} - \Pi$. 

\begin{theorem}\label{th.n-s-a-Robin}
Let $K$ be a compact subset of $\rho(H^{\eff})$. Then there exists $\eps_0 > 0$ and $C \geq 0$ such that for $z \in K$ and $\eps \in (0,\eps_0)$ we have $z \in \rho(H_{\eps})$ and
\[
\|(\PP_\eps-z)^{-1} - U_\eps ^{-1} (\LLeff - z)^{-1} \Pi U_\eps\|_{\mathcal{L}(L^2(\Omega_{\eps}))} \leq C \eps\,.
\]
Here $U_\eps$ is a unitary transformation from $L^2(\Omega_{\eps},\dx x)$ to $L^2(\Omega, w_\eps(x) \dx \sigma \dx t)$, where for some $C > 1$ we have 
\[
\forall \eps \in (0,\eps_0), \forall x \in \Omega, \quad \frac 1 C \leq|w_\eps(x)| \leq C\,.
\]
\end{theorem}

As for Theorem \ref{theo.main} it is implicit that the resolvent $(\LLeff-z)^{-1}$ is composed on the left by the inclusion $\Pi L^2(\O_\e) \to L^2(\O_\e)$. Moreover the operator $\LLeff$ on $L^2(\Sigma)$ has been identified with an operator on $\Pi L^2(\O_\e)$.

\begin{remark}
In the geometrically trivial situation $\Sigma=\R^{d-1}$
and special choice $\Re(\alpha)=0$,
a version of Theorem~\ref{th.n-s-a-Robin} was previously 
established in~\cite{BorisovKr12}. 
At the same time, in the self-adjoint case $\Im(\alpha)=0$
and very special geometric setting $d=1$ ($\Sigma$ being a curve),
a version of Theorem~\ref{th.n-s-a-Robin} is due to~\cite{Oliveira}. 
In our general setting, it is interesting to see
how the geometry enters the effective dynamics,
through the mean curvature of~$\Sigma$, see~\eqref{Veff}.
\end{remark}

\subsection{From variational estimates to norm resolvent convergence}

All the results of this paper are about estimates of the difference of resolvents of two operators. These estimates will be deduced from the corresponding estimates of the associated quadratic forms by the following general lemma:

\begin{lemma}\label{lem-diff-resolv}
Let $\Kc$ be a Hilbert space. Let $\AA$ and $\hAA$ be two closed densely defined operators on $\Kc$. Assume that $\hAA$ is bijective and that there exist $\eta_{1},\eta_{2},\eta_{3},\eta_{4}\geq 0$ such that $1 -\eta_{1} \|\hAA^{-1}\|-\eta_{2}>0$ and
\begin{multline*}
\forall \phi \in \Dom(\AA), \forall \psi \in \Dom(\hAA^*)\,,\\
\quad |\langle\AA\phi,\psi\rangle - \langle\phi,\hAA^* \psi\rangle| 
\leq \eta_1 \|\phi\|\|\psi\|+\eta_2 \|\phi\|\|\hAA^* \psi\|+\eta_{3}\|\AA \phi\|\|\psi\|+ \eta_{4} \| {\AA} \phi\|\|\hAA^* \psi\|\,.
\end{multline*}
Then $\AA$ is injective with closed range. If moreover $\AA^*$ is injective, then $\AA$ is bijective and we have the estimates 
\begin{equation}\label{A^-1.est}
\|\AA^{-1}\|\leq   \frac{ ( \eta_{3}+1) \| \hat {\AA}^{-1} \|+\eta_{4}}{1- \eta_1 \|\hAA^{-1}\|-\eta_2}
\end{equation}
and
\[
\left\| {\AA}^{-1}- \hAA^{-1}\right\|
\leq \eta_1 \|\AA^{-1} \| \|\hAA^{-1}\| + \eta_2 \|\AA^{-1} \|+\eta_{3}\|\hAA^{-1}\|+\eta_{4}.
\] 
\end{lemma}

Since the proof is rather elementary, let us provide it already now.
\begin{proof}
Let $\phi \in \Dom(\AA)$ and consider $\psi = (\hAA^{-1})^* \phi \in \Dom(\hAA^*)$. We have
\begin{align*}
|\|\phi\|^2- \langle\AA \phi,(\hAA^{-1})^*\phi\rangle|
& = |\langle\phi,\hAA^* \psi\rangle - \langle\AA\phi,\psi\rangle|\\
& \leq (\eta_1 \|\hAA^{-1}\|+\eta_{2})\|\phi\|^2+\left(\eta_{3}\|\hAA^{-1}\|+\eta_{4} \right)\|\AA\phi\|\|\phi\|\,,
\end{align*}
so
\[
\|\phi\|^2 \leq \left(\eta_1 \| \hat {\AA}^{-1} \|+ \eta_2 \right) \|\phi\|^2 + \left(( \eta_{3}+1) \| \hat {\AA}^{-1} \|+\eta_{4}\right)  \|\phi\| \|\AA \phi\|.
\]
Then if 
$
\eta_1 \|\hAA^{-1}\| + \eta_2 < 1\,,
$
we get
\begin{equation}\label{eq.lb-L2}
\|\phi\| \leq \frac{ ( \eta_{3}+1) \| \hat {\AA}^{-1} \|+\eta_{4}}{1- \eta_1 \|\hAA^{-1}\|-\eta_2} \|\AA\phi\|\,.
\end{equation}
In particular, ${\AA}$ is injective with closed range. If ${\AA}^*$ is injective, the range of ${\AA}$ is dense and thus ${\AA}$ is bijective. In particular, with \eqref{eq.lb-L2}, we obtain \eqref{A^-1.est}.

Finally for $f,g\in \Kc$, $\phi= \AA ^{-1} f$ and $\psi = (\hAA^{-1})^*g$ we have
\[
\langle\big( \AA^{-1} - \hAA^{-1} \big) f,g\rangle = \langle\phi,\hAA^* \psi\rangle - \langle\AA\phi,\psi\rangle,
\]
and the conclusion follows by easy manipulations.
\end{proof}

\subsection{Organization of the paper}
In Section \ref{sec-abstract}, we prove Theorem \ref{theo.main}. We first define the operators $\LL$, $\widehat{\mathscr{L}}$ and $\LLeff$, 
and then we show how Lemma \ref{lem-diff-resolv} can be applied. In Section \ref{sec-examples}, we discuss some applications of Theorem \ref{theo.main} to the semiclassical Born-Oppenheimer approximation, the Dirichlet Laplacian on a shrinking tubular neighborhood of an hypersurface and the Robin Laplacian in the large coupling limit. Section \ref{sec-Robin} is devoted to the proof of Theorem \ref{th.n-s-a-Robin} about the non-self-adjoint Robin Laplacian on a shrinking layer.

\section{Abstract reduction of dimension} \label{sec-abstract}

In this section we describe more precisely the setting introduced in Section \ref{sec-intro-abstract} and we prove Theorem \ref{theo.main}. The applications will be given in the following section.

\subsection{Definition of the effective operator}

Let $(\Sigma,\sigma)$ be a measure space. For each $s \in \Sigma$ 
we consider a separable complex Hilbert space $\mathcal{H}_s$.
Then, on $\Hc_s$ we consider a closed symmetric non-negative sesquilinear form $q_s$ with dense domain $\Dom(q_s)$. We denote by $T_s$ the corresponding self-adjoint and non-negative operator, as given by the Representation Theorem. As already said in Section~\ref{sec-intro-abstract}, we consider a function $s \in \Sigma \mapsto \gamma_s \in \R$ whose infimum is positive, see~\eqref{gamma}.
Then we denote by $\Pi_s \in \Lc(\Hc_s)$ the spectral projection of $T_s$ on $[0,\gamma_s)$, and we set $\Pi_s^\bot = \mathrm{Id}_{\Hc_s} - \Pi_s$.

We denote by $\Hc$ the subset of $\bigoplus_{s \in \Sigma} \Hc_s$ which consists of all $\Phi = (\Phi_s)_{s \in \Sigma}$ such that the functions $s \mapsto \nr{\Phi_s}_{\Hc_s}$ and $s \mapsto \nr{\Pi_s \Phi_s}_{\Hc_s}$ are measurable on $\Sigma$ and 
\[
\|\Phi\|^2 = \int_\Sigma \|\Phi_s\|^2_{\mathcal{H}_s} \dx\sigma(s) < + \infty\,.
\]
It is endowed with the Hilbert structure given by this norm. We denote by $\Pi$ the bounded operator on $\Hc$ such that for $\Phi \in \mathcal{H}$ and $s \in \Sigma$ we have $(\Pi \Phi)_s = \Pi_s \Phi_s$. We similarly define $\Pi^\bot \in \mathcal{L}(\mathcal{H})$. 

We say that $\Phi = (\Phi_s)_{s \in \Sigma} \in \mathcal{H}$ belongs to $\Dom(Q_T)$ if $\Phi_s$ belongs to $\Dom(q_s)$ for all $s \in \Sigma$, the functions $s \mapsto {q_s(\Phi_s)}$ and $s \mapsto q_s(\Pi_s \Phi_s)$ are measurable on $\Sigma$ and 
\[
Q_T(\Phi) = \int_\Sigma {q_s(\Phi_s)} \dx\sigma(s) < +\infty\,.
\]
We consider on $\Hc$ an operator $S$ with dense domain $\Dom(S)$. We assume that $\Dom(S)$ is invariant under $\Pi$, that $[S,\Pi]$ extends to a bounded operator on $\Hc$, and we define $a$ as in \eqref{a.def}.
We assume that
\[
\Dom(\Qo) = \Dom(S) \cap \Dom(Q_T)
\]
is dense in $\mathcal{H}$, and for $\Phi \in \Dom(\Qo)$ we set 
\begin{equation} \label{def-Q}
\Qo(\Phi) = \|S\Phi\|^2 + Q_T(\Phi)\,.
\end{equation}
We assume that $\Qo$ defines a closed form on $\mathcal{H}$. The form $\Qo$ is symmetric and non-negative and the associated operator is the operator $\LLo$ introduced in \eqref{def-LL}.

Then we define the operator $\hatLL$ (see \eqref{def-hatLL}) by its form. 
For this we need to verify that the form domain 
is left invariant both by~$\Pi$ and~$\Pi^\bot$.

\begin{lemma}\label{lem.Pi-DomQ}
For all $\Phi \in \Dom(\Qo)$ we have $\Pi \Phi \in \Dom(\Qo)$ and $\Pi^\bot \Phi \in \Dom(\Qo)$.
\end{lemma}

\begin{proof}
Let $\Phi = (\Phi_s)_{s \in \Sigma} \in \Dom(\Qo)$. We have $\Phi \in \Dom(S)$, so by assumption we have $\Pi \Phi \in \Dom(S)$. By assumption again, the function $s \mapsto q_s (\Pi_s \Phi_s) = q_s(\Pi_s \Pi_s \Phi_s)$ is measurable and we have 
\[
\int_{\Sigma} q_s(\Pi_s \Phi_s) \, d\sigma(s) \leq \sup_{s \in \Sigma} \gamma_s \int_{\Sigma} \nr{\Phi_s}_{\Hc_s}^2 \, d\sigma(s) < +\infty.
\]
This proves that $\Pi \Phi$ belongs to $\Dom(Q_T)$, and hence to $\Dom(\Qo)$. Then the same holds for $\Pi^\bot \Phi = \Phi - \Pi \Phi$.
\end{proof}

With this lemma we can set, for $\Phi,\Psi \in \Dom(\Qo)$,
\[
\hatQ (\Phi,\Psi) = \Qo(\Pi \Phi,\Pi \Psi) + \Qo(\Pi^\bot \Phi,\Pi^\bot \Psi)\,.
\]

\begin{lemma} \label{lem.hQ}
For all $\Phi \in\Dom(\hatQ)$ we have 
\[
\Qo(\Phi) \leq 2 \hatQ (\Phi)\,.
\]
In particular the form $\hatQ$ is non-negative, closed, and it determines uniquely a self-adjoint operator $\hatLL$ on $\Hc$. Moreover we have $[\Pi,\hatLL]=0$ on $\Dom(\hatLL)$.
\end{lemma}

\begin{proof}
We have 
\[
\Qo(\Phi) - \hatQ(\Phi) = \Qo(\Pi \Phi, \Pi^\bot \Phi) + \Qo(\Pi^\bot \Phi, \Pi \Phi)\,.
\]
Since the form $\Qo$ is non-negative we can apply the Cauchy-Schwarz inequality to write 
\[
\Qo(\Pi \Phi, \Pi^\bot \Phi) 
\leq 
\sqrt {\Qo(\Pi \Phi)} \mbox{$\sqrt{\Qo(\Pi^\bot \Phi)}$}
\leq 
\frac 12 \big( \Qo(\Pi \Phi) + \Qo(\Pi^\bot \Phi) \big) = \frac 12 \hatQ(\Phi)\,.
\]
We have the same estimate for $\Qo(\Pi^\bot \Phi, \Pi \Phi)$, and the first conclusions follow. We just check the last property about the commutator. Let $\psi\in\Dom(\hatLL)$. For all $\phi \in \Dom(\hatLL)$ we have
\[
\hatQ(\phi,\Pi\psi)=\Qo(\Pi\phi,\Pi\psi)=\hatQ(\Pi\phi,\psi)=\langle\Pi\phi,\hatLL\psi\rangle_\Hc = \langle \phi, \Pi \hatLL \psi \rangle_\Hc.
\]
This proves that $\Pi\psi\in\Dom(\hatLL)$ with $\hatLL\Pi\psi = \Pi \hatLL \psi$ and the proof is complete.
\end{proof}

Then, from $\hatQ$ it is easy to define the forms corresponding to the operators $\LLeff$ and $\LLbot$:

\begin{lemma}\label{lem.QPi}
Let $\Qeff$ be the restriction of $\Qo$ to $\Pi\Dom(\Qo) = \mathrm{Ran}(\Pi) \cap \Dom(\Qo)$. Then $\Qeff$ is non-negative and closed. The associated operator $\LLeff$ is self-adjoint, its domain is invariant under $\Pi$, and $[\Pi,\LLeff]=0$ on $\Dom(\LLeff)$. Moreover, we have $(\Dom(\hatLL)\cap\mathrm{Ran}(\Pi),\hatLL)=(\Dom(\LLeff),\LLeff)$.

We have similar statements for the restriction $\Qbot$ of $\Qo$ to $\Pi^\bot \Dom(\Qo) = \mathrm{Ran}(\Pi^\bot) \cap \Dom(\Qo)$ and the corresponding operator $\LLbot$.
\end{lemma}

\begin{proof}
The closedness of $\Qeff$ comes from the closedness of $\Qo$ and the continuity of $\Pi$. The other properties are proved as for Lemma \ref{lem.hQ}. We prove the last assertion. Let $\psi\in\Dom(\LLeff)$. By definition of this domain we have $\Pi\psi=\psi$. For $\phi\in\Dom(\hatQ)$, we have 
\[
\hatQ(\phi,\psi)=\Qo(\Pi\phi,\Pi\psi)=\Qeff(\Pi\phi,\Pi\psi)=\Qeff(\Pi\phi,\psi) = \innp{\Pi\phi}{\LLeff \psi} = \innp{\phi}{\LLeff \psi}. 
\]
This proves that $\psi\in\Dom(\hatLL)$ and $\LLeff\psi=\hatLL\psi$. Thus $\Dom(\LLeff) \subset \Dom(\hatLL)\cap\mathrm{Ran}(\Pi)$ and $\hatLL = \LLeff$ on $\Dom(\LLeff)$. The reverse inclusion $\Dom(\hatLL)\cap\mathrm{Ran}(\Pi) \subset \Dom(\LLeff)$ is easy, so the proof is complete.
\end{proof}
Finally we have proved that 
\[
\Dom(\hatLL) = \big(\Dom(\hatLL) \cap \mathrm{Ran}(\Pi)\big) \oplus \big(\Dom(\hatLL) \cap \mathrm{Ran}(\Pi^\bot)\big) = \Dom(\LLeff) \oplus \Dom(\LLbot)
\]
and for $\f \in \Dom(\hatLL)$ we have 
\[
\hatLL \f = \LLeff \Pi \f + \LLbot \Pi^\bot \f.
\]
From the spectral theorem, we deduce the following lemma.
\begin{lemma} \label{lem-res-bot}
We have $\Sp(\hatLL) = \Sp(\LLeff) \cup \Sp(\LLbot)$ and, for $z \in \rho(\hatLL)$ such that $z\notin [\gamma,+\infty)$,
\[
\nr{(\hatLL-z)^{-1} - (\LLeff-z)^{-1} \Pi} \leq \frac 1 {\mathsf{dist}(z, [\gamma,+\infty))}\,.
\]
\end{lemma}

\subsection{Comparison of the resolvents}
This section is devoted to the proof of the following theorem that implies Theorem \ref{theo.main} via Lemma \ref{lem-diff-resolv}.
\begin{theorem}\label{th-comp-forms}
Let $\LLo$ and $\hatLL$ be as above. Let $z \in \C$ and $\eta_1(z),\eta_2(z),\eta_3(z),\eta_4(z)$ as in \eqref{def-eta}. Then for $\Phi \in \Dom(\LLo)$ and $\Psi \in \Dom(\hatLL^*)$ we have 
\begin{align*}
\abs{\Qo(\Phi,\Psi) -\hatQ(\Phi,\Psi)}
& \leq \eta_1(z) \nr{\Phi}\nr{\Psi} + \eta_2(z) \|\Phi\|  \|(\hatLL-\bar z) \Psi\|\\
& + \eta_3(z) \|(\LLo-z) \Phi\|  \|\Psi\|  + \eta_4(z) \|(\LLo-z)\Phi\|  \|(\hatLL-\bar z) \Psi\|\,.
\end{align*}
\end{theorem}
Theorem \ref{th-comp-forms} is a consequence of the following proposition after inserting $z$ and using the triangular inequality.
\begin{proposition}\label{prop.diff.Q0Q}
For all $\Phi \in \Dom(\LLo)$ and $\Psi \in \Dom(\hatLL)$ we have 
\begin{multline*}
\frac 1 \gamma |\Qo(\Phi,\Psi) - \hatQ(\Phi,\Psi)| \\
\leq \frac{3 a}{\sqrt{2}} \left( \|\Phi\| +  \frac {\|\LLo \Phi\|}{\gamma} \right) \frac {\|\hatLL \Psi \|}{\gamma} + \frac{3 a}{\sqrt{2}} \left( a \|\Phi\| +\left(1+\frac{a}{\sqrt{2}}\right)  \frac {\| \LLo \Phi \|}{\gamma} \right) \left( \|\Psi\| +  \frac {\|\hatLL \Psi \|}{\gamma} \right)\,.
\end{multline*}
\end{proposition}

\begin{proof}
Let $\nu = \nr{[S,\Pi]}$. We have
\[\Qo(\Phi,\Psi) - \hatQ(\Phi,\Psi)=\Qo(\Pi^\bot\Phi,\Pi\Psi)+\Qo(\Pi\Phi,\Pi^\bot\Psi)\,.\]
For the first term we write
\[\Qo(\Pi^\bot\Phi,\Pi\Psi)=\langle S\Pi^\bot\Phi,S\Pi\Psi\rangle=\langle S\Pi^\bot\Phi,[S,\Pi]\Pi\Psi\rangle+\langle S\Pi^\bot\Phi,\Pi S\Pi\Psi\rangle\,,\]
so that
\[\Qo(\Pi^\bot\Phi,\Pi\Psi)=\langle S\Pi^\bot\Phi,[S,\Pi]\Pi\Psi\rangle+\langle [S,\Pi^\bot]\Pi^\bot\Phi,\Pi S\Pi\Psi\rangle\,.\]
We deduce that
\begin{equation}\label{eq.mix-term-a}
|\Qo(\Pi^\bot\Phi,\Pi\Psi)|\leq \nu\|S\Pi^\bot\Phi\|\|\Psi\|+\nu\|\Pi^\bot\Phi\|\|S\Pi\Psi\|\,.
\end{equation}
Similarly, we get, by slightly breaking the symmetry,
\begin{equation}\label{eq.mix-term-b}
|\Qo(\Pi\Phi,\Pi^\bot\Psi)|\leq \nu\|S\Pi^\bot\Psi\|\|\Phi\|+\nu\|\Pi^\bot\Psi\|\|S\Phi\|\,.
\end{equation}
We infer that
\begin{equation}\label{eq.diff-forms0}
|\Qo(\Phi,\Psi) - \hatQ(\Phi,\Psi)| \leq  \nu\|S\Pi^\bot\Phi\|\|\Psi\|+\nu\|\Pi^\bot\Phi\|\|S\Pi\Psi\|+ \nu\|S\Pi^\bot\Psi\|\|\Phi\|+\nu\|\Pi^\bot\Psi\|\|S\Phi\|\,.
\end{equation}
Since $Q_T$ is non-negative we have 
\begin{equation}\label{eq.SPhi}
\|S\Phi\|^2\leq \Qo(\Phi)\leq\|\LLo \Phi\|\|\Phi\|.
\end{equation}
Similarly,
\begin{equation}\label{eq.SPiPsi}
\|S\Pi\Psi\|^2\leq\hatQ(\Psi)\leq\|\hatLL\Psi\|\|\Psi\|\,.
\end{equation}
Then we estimate $\|\Pi^\bot\Phi\|$ and $\|S\Pi^\bot\Phi\|$. We have
\[\langle\Pi^\bot\Phi, \LLo \Phi\rangle= \Qo(\Pi^\bot\Phi,\Phi)=\Qo(\Pi^\bot\Phi)+\Qo(\Pi^\bot\Phi,\Pi\Phi)\,,\]
and deduce
\[\Qo(\Pi^\bot\Phi)\leq \|\LLo \Phi\|\|\Pi^\bot\Phi\|+|\Qo(\Pi^\bot\Phi, \Pi\Phi)|\,.\]
From \eqref{eq.mix-term-b}, we get
\[\Qo(\Pi^\bot\Phi)\leq \|\LLo \Phi\|\|\Pi^\bot\Phi\|+ \nu\|S\Pi^\bot\Phi\|\|\Phi\|+\nu\|\Pi^\bot\Phi\|\|S\Phi\|\,.\]
Moreover, we have
\[\Qo(\Pi^\bot\Phi)\geq \|S\Pi^\bot\Phi\|^2+\gamma\|\Pi^\bot\Phi\|^2\,.\]
We infer that
\begin{multline*}
 \|S\Pi^\bot\Phi\|^2+\gamma\|\Pi^\bot\Phi\|^2\\
\leq
\frac{\gamma}{4}\|\Pi^\bot\Phi\|^2+\frac{1}{\gamma}\|\LLo \Phi\|^2+ \frac{1}{2}\|S\Pi^\bot\Phi\|^2+\frac{\nu^2}{2}\|\Phi\|^2+\frac{\gamma}{4}\|\Pi^\bot\Phi\|^2+\frac{\nu^2}{\gamma}\|S\Phi\|^2\,.
\end{multline*}
Using \eqref{eq.SPhi} we deduce that
\[ 
\frac{1}{2}\left(\|S\Pi^\bot\Phi\|^2+\gamma\|\Pi^\bot\Phi\|^2\right)
\leq \frac{1}{\gamma}\|\LLo \Phi\|^2+\frac{\nu^2}{2}\|\Phi\|^2+\frac{\nu^2}{2}\left(\frac{\|\LLo \Phi\|^2}{\gamma^2}+\|\Phi\|^2\right),
\]
and thus
\begin{equation}\label{eq.control-PibotPhi}
\frac {\|S\Pi^\bot\Phi\|^2} \gamma  + \|\Pi^\bot\Phi\|^2 \leq (2+a^2) \frac {\|\LLo \Phi\|^2}{\g^2} + 2a^2 \|\Phi\|^2\,.
\end{equation}
Let us now consider $\|\Pi^\bot\Psi\|$ and $\|S\Pi^\bot\Psi\|$. We have easily that
\[\|S\Pi^\bot\Psi\|^2+\gamma\|\Pi^\bot\Psi\|^2\leq  \Qo(\Pi^\bot\Psi)=\hatQ(\Psi,\Pi^\bot\Psi)\leq\|\hatLL\Psi\|\|\Pi^\bot\Psi\|\,,\]
and thus
\begin{equation}\label{eq.control-PibotPsi}
\frac {\|S\Pi^\bot\Psi\|^2} \gamma + \|\Pi^\bot\Psi\|^2\leq\frac{\|\hatLL\Psi\|^2}{\gamma^2}.
\end{equation}
It remains to combine \eqref{eq.diff-forms0}, \eqref{eq.SPhi}, \eqref{eq.SPiPsi}, \eqref{eq.control-PibotPhi}, \eqref{eq.control-PibotPsi}, and use elementary manipulations.
\end{proof}  

\section{Examples of applications} \label{sec-examples}
In this section we discuss three applications of Theorem \ref{theo.main} and we recall that we are in the context of Remark \ref{rem.0}.

\subsection{Semiclassical Born-Oppenheimer approximation}
In this first example we set $(\Sigma,\sigma)=(\R,\dx s)$. 
We consider a Hilbert space $\Hc_T$ and set $\Hc = L^2(\R,\Hc_T)$. Then, for $h > 0$, we consider on $\Hc$ the operator $S_h = h D_{s}$, where $D_s = -i\partial_s$. We also consider an operator $T$ on $\Hc$ such that for $\Phi = (\Phi_s)_{s \in \R} \in \Hc$ we have $(T \Phi)_s = T_{s} \Phi_s$, where $(T_{s})$ is a family of operators on $\Hc_T$ which depends analytically on $s$. Thus the operator $\LLo = \LLh$ takes the form
\[
\LLh=h^2D^2_{s}+T\,.
\]
This kind of operators appears in \cite{LT13, Lampart2016} where their spectral and dynamical behaviors are analyzed. As an example of operator $T$, the reader can have the Schr\"odinger operator $-\Delta_{t}+V(s,t)$ in mind,
where the electric potential~$V$ is assumed to be real-valued.
Here the operator norm of the commutator $[hD_{s},\Pi]$ is controlled by the supremum of $\|\partial_{s}u_{1}(s)\|_{\mathcal{H}}$. Assuming that $\|\partial_{s}u_{1}(s)\|_{\mathcal{H}}$ is bounded, we have $a=a(h)=\mathcal{O}(h)$ (see \eqref{a.def}). Let us also assume, for our convenience, that $\mu_{1}$ has a unique minimum, non-degenerate and not attained at infinity. Without loss of generality we can assume that this minimum is 0 and is attained at $0$. Thus, here $\gamma$ just satisfies $\gamma=\inf_{s\in\R}\mu_{2}(s)>0$.

For $k \in \N^*$ we set 
\begin{equation} \label{def-lambda}
\l_k(h) = \sup_{\substack {F \subset \Dom(\LLh)  \\ \mathrm{codim}(F) = k-1}} \inf_{\substack{\f\in F \\ \nr{\f} = 1}} \innp{\LLh \f}{\f}.
\end{equation}
By the min-max principle, the first values of $\l_k(h)$ are given by the non-decreasing sequence of isolated eigenvalues of $\LLh$ (counted with multiplicities) below the essential spectrum. If there is a finite number of such eigenvalues, the rest of the sequence is given by the minimum of the essential spectrum. We similarly define the sequence $(\l_{\eff,k}(h))$ corresponding to the operator $\LLeffh$. Note that $\LLeffh$ can be identified with the operator
\[h^2D_{s}^2+\mu_{1}(s)+h^2\|\partial_{s}u_{1}(s)\|_{\mathcal{H}_T}^2\,.\]
As a consequence of the harmonic approximation (see for instance \cite[Chapter 7]{FouHel10} or \cite[Section 4.3.1]{Ray17}), we get the following asymptotics.
\begin{proposition} \label{prop-spectre-op-reduit}
Let $k\in\N^*$. We have
\[
  \l_{\eff,k}(h)=(2k-1)\sqrt{\frac{\mu''(0)}{2}}h+o(h)
  \,, \qquad h \to 0
  \,.
\]

\end{proposition}

From our abstract analysis, we deduce the following result.

\begin{proposition}\label{cor.diff-res.BO}
Let $c_0,C_0 > 0$. There exist $h_{0}>0$ and $C>0$ such that for $h\in(0,h_{0})$ and 
\[
z \in \Zc_h = \left\{z \in [-C_0 h,C_0 h] \, : \, \mathsf{dist}(z,\Sp(\LLeffh)) \geq c_0 h\right\}
\]
we have $z \in \rho(\LLh)$ and
\[
\|(\LLh-z)^{-1}-(\LLeffh-z)^{-1}\|\leq C\,.
\]
\end{proposition}

\begin{proof}
Let $h > 0$ and $z \in \Zc_{h}$. If $h$ is small enough we have $C_0 h < \gamma$ so $z \in \rho(\LLeffh) \cap \rho(\LLbot_h) = \rho(\hatLL_h)$. Moreover, by the Spectral Theorem,
\[
\nr{(\hatLL_h-z)^{-1}} \leq \nr{(\LLeffh-z)^{-1}} + \nr{(\LLbot_h-z)^{-1}} \leq \frac 1 {c_0 h} + \frac 1 {\gamma - C_0 h}.
\]
With the notation \eqref{def-eta} we have  
\[
\liminf_{h \to 0} \sup_{z \in Z_{h}} \left(1 - \eta_{1,h}(z)\|(\hatLL_h-z)^{-1}\|-\eta_{2,h}(z) \right) > 0.
\]
From Theorems \ref{theo.main} and Proposition \ref{prop.compl-theo.main}, we deduce that $z\in\rho(\LLh)$, 
\[\|(\LLh-z)^{-1}\|\lesssim h^{-1}\,,\]
and the estimate on the difference of the resolvents.
Here and occasionally in the sequel,
we adopt the notation $x \lesssim y$ if there is 
a positive constant~$C$ (independent of~$x$ and~$y$) 
such that $x \leq C y$.
\end{proof}

From this norm resolvent convergence result, we recover a result of \cite[Section 4.2]{Lampart2016}.
\begin{proposition} \label{prop-eig-BH}
Let $k \in \N^*$. 
Then
\[
\l_k(h) = \l_{\eff,k}(h) +  \mathcal{O}(h^2), \quad h \to 0.
\]
\end{proposition}

\begin{proof}
Let $\e > 0$ be such that $\l_{\eff,k+1}(h) - \l_{\eff,k}(h) > 2\e h$  for all $h$. We set $z_h = \l_{\eff,k}(h) + \e h$. The resolvent $(\LLeffh-z_h)^{-1}$ has $k$ negative eigenvalues
\[
\frac 1 {\l_{\eff,k}(h)-z_h} \leq \dots \leq \frac 1 {\l_{\eff,1}(h)-z_h},
\]
all smaller than $-\a/h$ for some $\a > 0$, and the rest of the spectrum is positive. By Proposition~\ref{cor.diff-res.BO} the resolvent ${(\LLh-z_h)^{-1}}$ is well defined for $h$ small enough and there exists $C > 0$ such that 
\[
\nr{(\LLh-z_h)^{-1} - (\LLeffh-z_h)^{-1}} \leq C.
\]
By the min-max principle applied to these two resolvents, we obtain that for all $j \in \{ 1,\dots, k\}$ the $j$-th eigenvalue of $(\LLh-z_h)^{-1}$ is at distance not greater than $C$ from $1 / (\l_{\eff,k+1-j}-z_h)$, and the rest of the spectrum is greater than $-C$. In particular, for $j = 1$,
\[
\abs{\frac 1 {\l_k(h)-z_h} - \frac 1 {\l_{\eff,k}(h)-z_{h}}} \leq C.
\]
This gives
\[
\abs{\l_k(h) - \l_{\eff,k}(h)} \leq  C \e h \abs{\l_k(h) - \l_{\eff,k}(h) + \e h} ,
\]
and the conclusion follows for $h$ small enough.
\end{proof}

\subsection{Shrinking neighborhoods of hypersurfaces}\label{sec.shrink.Dir}

In this paragraph we consider a submanifold $\Sigma$ of $\R^d$, $d \geq 2$, as in Section \ref{sec-intro-nsa-robin}. We choose $\e > 0$ and define $\Th_\e$, $\O$ and $\O_\e$ as in \eqref{def-Theta} and \eqref{def-Omega-eps}.
For $\varphi\in H^1_{0}(\Omega_{\eps})$, we set
\[
Q^\mathsf{Dir}_{\O_\e}(\varphi)=\int_{\Omega_{\eps}}|\nabla\varphi|^2\dx x,
\]
and we denote by $-\Delta_{\Omega_{\eps}}^\mathsf{Dir}$ the associated operator.
Then we use the diffeomorphism $\Theta_\eps$ to see $-\Delta^\mathsf{Dir}_{\Omega_{\eps}}$ as an operator on $L^2(\Omega)$. We set, for $\psi \in H^1_{0}(\Omega, \dx \sigma  \dx t)$, 
\[
\mathcal{Q}_\eps^{\mathsf{Dir}}(\psi) = Q_{\O_\eps}^{\mathsf{Dir}}(\psi \circ \Theta_\eps^{-1})\,.
\]

We need a more explicit expression of $\mathcal{Q}^{\mathsf{Dir}}_\eps$ in terms of the variables $(s,t)$ on $\Omega$. For $(s,t) \in \Omega$ we have on $T_{(s,t)} \Omega \simeq T_s\Sigma \times n(s) \R$ 
\[
d_{(s,t)} \Theta_\eps = (\mathsf{Id}_{T_s\Sigma} + \eps t d_s n) \otimes \eps \mathsf{Id}_{n(s) \R}\,.
\]
Hence 
\[
d_{\Theta_\eps(s,t)} \Theta_\eps^{-1} = (\mathsf{Id}_{T_s\Sigma} + \eps t d_s n)^{-1} \otimes \eps^{-1} \mathsf{Id}_{n(s) \R}\,.
\]
We recall that the Weingarten map $-d_s n$ is a self-adjoint operator on $T_s\Sigma$ (endowed with the metric inherited from the Euclidean structure on $\R^d$). For $\psi \in H^1(\Omega, \dx \sigma  \dx t)$, $x \in \Omega_{\eps}$ and $(s,t) = \Theta_\eps^{-1}(x)$ we get 
\begin{align*}
\|\nabla (\psi \circ \Theta_\eps^{-1})(x)\|_{T_x \O_\e}^2
& = \|(d_x \Theta_\eps^{-1})^* \nabla \psi(s,t)\|_{T_x \O_\e}^2\\
& = \|(\mathsf{Id}_{T_s\Sigma} + \eps t d_s n)^{-1} \nabla_s \psi(s,t)\|_{T_s \Sigma}^2 + \frac 1 {\eps^2} \abs{\partial_t \psi (s,t)}^2\,.
\end{align*}
The eigenvalues of the Weingarten map are the principal curvatures $\kappa_1,\dots,\kappa_{d-1}$. In particular for $(s,t)\in\Omega$ we have 
\begin{equation} \label{det-w}
|d_{(s,t)}\Theta_\eps| = \eps {w}_{\eps}, \quad \text{where }  w_\eps(s,t) = \prod_{j=1}^{d-1} (1- \eps t \kappa_j(s))\,.
\end{equation}
The Riemannian structure on $\Omega$ is given by the pullback by $\Theta_\eps$ of the Euclidean structure defined on $\Omega_{\eps}$. More explicitly, for $(s,t) \in \Omega$ the inner product on $T_{(s,t)} \O$ is given by
\[
\forall X ,Y  \in T_{(s,t)} (\Omega), \quad  g_\eps(X,Y) = \langle d_{(s,t)} \Theta_\eps (X), d_{(s,t)} \Theta_\eps (Y)\rangle_{\R^d}\,.\]
Then the measure corresponding to the metric $g_\eps$ is given by $\eps {w}_{\eps} \dx \sigma\dx t$. Thus, if we set 
\begin{equation} \label{def-G}
G_\eps(s,t) = (\mathsf{Id}_{T_s\Sigma} + \eps t d_s n)^{-2},
\end{equation}
we finally obtain
\begin{align*}
\mathcal{Q}_\eps^{\mathsf{Dir}}(\psi)
& = \int_{\Omega_{\eps}} |(\mathsf{Id}_{T_s\Sigma} + \eps t d_s n)^{-1} \nabla_s \psi (\Theta_\eps^{-1} (x))|^2 \dx x + \frac 1 {\eps^2} \int_{\Omega_{\eps}} |\partial_t \psi (\Theta_\eps^{-1} (x))|^2 \dx x \\
& = \eps \int_{\Omega} \langle G_\eps(s,t) \nabla_s \psi, \nabla_s \psi\rangle_{T\Sigma} {w}_{\eps} \dx\sigma\dx t  + \frac{1}{\eps^2} \int_{\Omega} |\partial_t \psi|^2 \eps {w}_{\eps} \dx\sigma\dx t\,.
\end{align*}
The transverse operator $T_{s}(\e)$ is the Dirichlet realization on $L^2((-1,1),\eps  w_{\eps}\dx t)$ of the differential operator $-\eps^{-2} w^{-1}_{\eps}\partial_{t} w_{\eps}\partial_{t}$. We denote by $\m_1(s,\e)$ its first eigenvalue and we set $\mu(\eps)=\inf_{s\in\R}\mu_{1}(s,\eps)$. We have, by perturbation theory, as $\eps \to 0$,
\[
\mu_{1}(s,\eps)=\frac{\pi^2}{4\eps^2}+V(s)+\mathcal{O}(\eps)\,,\quad\mu(\eps)=\frac{\pi^2}{4\eps^2}+\mathcal{O}(1),
\]
where
\[
V(s)=-\frac{1}{2}\sum_{j=1}^{d-1}\kappa_{j}(s)^2+\frac{1}{4}\left(\sum_{j=1}^{d-1}\kappa_{j}(s)\right)^2.
\]
We denote by $\LLDir$ the operator associated to the form $\mathcal{Q}_\eps^{\mathsf{Dir}}$ and by $\LLDireff$ the corresponding effective operator as defined in the general context of Section \ref{sec.general-reduction}. It is nothing but the operator associated with the form $H^1(\Sigma)\ni\varphi\mapsto\mathcal{Q}_\eps^{\mathsf{Dir}}(\varphi u_{s,\eps})$ where $u_{s,\eps}$ is the positive $L^2$-normalized groundstate of the transverse operator (and actually depending on the principal curvatures analytically). From perturbation theory, we can easily check that the commutator between the projection on $u_{s,\eps}$ and $S$ is bounded (and of order $\eps$).

\begin{proposition}\label{cor.waveguides}
Let $c_0,C_0 > 0$. There exist $\eps_{0}>0$ and $C>0$ such that for all $\eps\in(0,\eps_{0})$ and
\[z \in \Zc_{c_{0}, C_{0},\eps}=\{z\in\R : |z-\mu(\eps)|\leq C_{0}\,,\quad\mathsf{dist}(z,\Sp(\LLDireff))\geq c_{0}\}\]
we have
\[\left\|\left(\LLDir-z\right)^{-1}-\left(\LLDireff-z\right)^{-1}\right\|\leq C\eps\,.\]
\end{proposition}
We recover a result of \cite{KRT15} (when there is no magnetic field). 
\begin{proof}
We are in the context of Remark \ref{rem.0}. The form $\Qc_\e - \m(\e)$ is non-negative. We denote by $\LLe$ the corresponding non-negative self-adjoint operator and define $\hatLLe$ as in Lemma \ref{lem.hQ}. Given $\e > 0$ and $z \in \Zc_{c_{0}, C_{0},\e}$ we write $\z$ for $z - \m(\e)$. Thus, with the notation of the abstract setting we have $\g_\e \sim \e^{-2}$, $a_\e = \mathcal{O}(\e^2)$, $\z = \mathcal{O}(1)$ and hence $\eta_{1,\e}(\z)=\mathcal{O}(\eps)$, $\eta_{2,\e}(\z)=\mathcal{O}(\eps^2)$, $\eta_{3,\e}(\z)=\mathcal{O}(\eps)$ and $\eta_{4,\e}(\z)=\mathcal{O}(\eps^2)$. Moreover, by the spectral theorem, we have
\[
\left\|\big(\hatLLe-\z\big)^{-1}\right\|=\mathcal{O}(1).
\]
Thus, there exists $\eps_{0}>0$ such that for $\eps\in(0,\eps_{0})$, $z\in Z_{c_{0},C_{0},\eps}$ and $\z = z-\m(\e)$ the operator $\LLe-\z$ is bijective and
\[\left\|\left(\LLe-\z\right)^{-1}\right\|=\mathcal{O}(1)\,,\]
\[\left\|\left(\LLe-\z\right)^{-1}-\big(\hatLLe-\z\big)^{-1}\right\|=\mathcal{O}(\eps)\,.\]
The conclusion easily follows.
\end{proof}

Given $\e > 0$ we define the sequence $(\l^{\mathsf{Dir}}_{k}(\e))_{k \in \N^*}$ and $(\l^{\mathsf{Dir}}_{k,\eff}(\e))_{k \in \N^*}$ corresponding to the operators $\LLDir$ and $\LLDireff$ as in \eqref{def-lambda}. By using analytic perturbation theory with respect to the parameters $(\eps\kappa_{j})_{1\leq j\leq d-1}$ to treat the commutator, we have, for all $k \in \N^*$, 
\[
\l^{\mathsf{Dir}}_{k,\eff}(\e) = \frac{\pi^2}{4\eps^2}+\l_{k}^\Sigma+\mathcal{O}(\eps)\,, \quad \eps \to 0,
\]
where $\l_{k}^\Sigma$ is the $k$-th eigenvalue of $-\Delta_{s}+V(s)$. 

We recover a result in the spirit of \cite{Duclos95, KRT15}.
\begin{proposition}\label{cor.waveguides-ev}
For all $k\geq 1$ we have
\[
\lambda^\mathsf{Dir}_{k}(\eps)=\frac{\pi^2}{4\eps^2}+\l_{k}^\Sigma+\mathcal{O}(\eps), \quad \eps \to 0
\]
\end{proposition}

\begin{proof}
Let $k\geq 1$. There exist $c_{0},\tilde c_{0}, C_{0},\eps_{0}>0$ such that for $\e \in (0,\e_0)$ we have
\[
\lambda^{\mathsf{Dir}}_{k,\eff}(\e)+\tilde c_{0} \in \Zc_{c_{0}, C_{0}, \eps}.
\]
As in the proof of Proposition \ref{prop-eig-BH} we obtain from Proposition \ref{cor.waveguides} and the min-max principle
\[\left|\left(\lambda^{\mathsf{Dir}}_{k}(\e)-(\lambda^{\mathsf{Dir}}_{k,\eff}(\e)+\tilde c_{0})\right)^{-1}-\left(\lambda^{\mathsf{Dir}}_{k,\eff}(\e)-(\lambda^{\mathsf{Dir}}_{k,\eff}(\e)+\tilde c_{0})\right)^{-1}\right|=\mathcal{O}(\eps).\]
We deduce
\[\left|\lambda^{\mathsf{Dir}}_{k}(\e)-\lambda^{\mathsf{Dir}}_{k,\eff}(\e)\right|=\mathcal{O}(\eps)\left|\left(\lambda^{\mathsf{Dir}}_{k}(\e)-(\lambda^{\mathsf{Dir}}_{k,\eff}(\e)+\tilde c_{0})\right)\right|\,,\]
and the conclusion follows.
\end{proof}

\subsection{Dirichlet-Robin shell with large coupling constant}
In this section, we keep considering the hypersurface $\Sigma$ of the last paragraph (here $\eps=1$). Let us now consider the Dirichlet-Robin Laplacian in an annulus. In other words, with $w_1$ and $G_1$ as defined by \eqref{det-w} and \eqref{def-G}, we consider on the weighted space $L^2(w_{1}\dx s\dx t)$ the quadratic form
\[\mathcal{Q}^{\mathsf{DR}}_{\alpha}(\psi)=\int_{\Sigma\times (0,1)}\left(\langle G_{1}(s,t)\nabla_{s}\psi,\nabla_{s}\psi\rangle_{T\Sigma}+ |\partial_{t}\psi|^2\right) w_{1}(s,t)\dx s\dx t-\alpha\int_{\Sigma}|\psi(s,0)|^2\dx s.\]
It is defined for $\psi\in \mathsf{Dom}(\mathcal{Q}^{\mathsf{DR}}_{\alpha})$ where
\[ \mathsf{Dom}(\mathcal{Q}^{\mathsf{DR}}_{\alpha})=\{\psi\in H^1(\Sigma\times(0,1)) :  \psi(s,1)=0\,,\quad \partial_{t}\psi(s,0)=-\alpha\psi(s,0)\}\,.\]

In these definitions $\alpha$ is real, and we are interested in the strong coupling limit $\alpha\to+\infty$.

This quadratic form is of the form \eqref{def-Q} with $S=G_{1}^{\frac{1}{2}}\nabla_{s}$ and $T_s=-w_{1}^{-1}\partial_{t}w_{1}\partial_{t}$ acting on $H^2((0,1))$ and Dirichlet-Robin condition. The spectrum of $T_s$ is well-understood in the limit $\alpha\to+\infty$. Actually, the family $(T_s)$ depends analytically on the principal curvatures $(\kappa_{j}(s))_{1\leq j\leq d-1}$. We can deduce from the previous works \cite{HelKac15, HelKacRay15, KKR16} that, as $\alpha \to + \infty$,
\[\mu_{1}(s,\alpha)=-\alpha^2-\alpha\kappa(s)+\mathcal{O}(1)\,,\quad\mu_{2}(s,\alpha)\geq c>0\,,\]
and
\[\mu(\alpha)=\inf_{s\in\Sigma}\mu_{1}(s,\alpha)=-\alpha^2-\alpha\kappa_{\max}+\mathcal{O}(1)\,,\]
with $\kappa=\sum_{j=1}^{d-1}\kappa_{j}$. 
Here, for simplicity, we assume that~$\kappa$ has a unique maximum at $s=0$
that is not degenerate and not attained at infinity. Moreover, we assume that the eigenvalues of $D^2_{s}+\frac{1}{2}\mathsf{Hess}_{0} (-\kappa)(s,s)$ are simple.
We let
\[\Zc_{C_{0}, c_{0},\alpha}=\{z\in\R : |z-\mu(\alpha)|\leq c_{0}\alpha\,,\quad \mathrm{dist}(z,\Sp(\hatLLDR)) \geq  C_{0}\}\,.\]

\begin{proposition}\label{cor.s-a-Robin}
There exist $C,\alpha_{0}>0$ such that, for all $z\in \Zc_{C_{0}, c_{0},\alpha}$
\[\left\|\left(\mathscr{L}_{\alpha}^{\mathsf{DR}}-z\right)^{-1}-\left(\LLDReff-z\right)^{-1}\right\|\leq C\alpha^{-1}\,.\]
\end{proposition}

\begin{proof}
Here we have $\gamma= \mathcal{O}(\alpha^2)$, $\nu=\mathcal{O}(\alpha^{-1})$ and $a=\mathcal{O}(\alpha^{-2})$. We use again Remark \ref{rem.0} and we apply Theorem \ref{theo.main} with $\LLp=\mathscr{L}^\mathsf{DR}_{\alpha}-\mu(\alpha)$ and $z$ replaced by $z-\mu(\alpha)$. For $z\in \Zc_{c_{0}, C_{0},\alpha}$, we get 
\[\eta_{1}=\mathcal{O}(\alpha^{-1})\,,\quad \eta_{2}=\mathcal{O}(\alpha^{-2})\,,\quad \eta_{3}=\mathcal{O}(\alpha^{-2})\,,\quad \eta_{4}=\mathcal{O}(\alpha^{-3})\,.\]
Moreover, for $\alpha$ large enough, we have, for all $z\in Z_{c_{0}, C_{0},\alpha}$, $z\in\rho\left(\LLDReff\right)$ and
\[\|(\LLDReff-z)^{-1}\|\leq C\,.\]
Then Theorem \ref{theo.main} implies the wished estimate.
\end{proof}

We recover, under our simplifying assumptions, a result appearing in \cite{HelKac15, PP-eh, KKR16}.

\begin{proposition}\label{cor.s-a-Robin-ev}
For all $j\geq 1$, we have, as $\alpha \to + \infty$,
\[\lambda^\mathsf{DR}_{j,\eff}(\alpha)=-\alpha^2+\nu_{j}(\alpha)+\mathcal{O}(1)\,,\]
and
\[\lambda^\mathsf{DR}_{j}(\alpha)=-\alpha^2+\nu_{j}(\alpha)+\mathcal{O}(1)\,,\]
where $\nu_{j}(\alpha)$ is the $j$-th eigenvalue of $D^2_{s}-\alpha\kappa(s)$.
\end{proposition}

\begin{proof}
Let us first discuss the asymptotic behavior of the eigenvalues of the effective operator. Let us recall that it is defined as explained in Section \ref{sec.general-reduction}, and that it can be identified with the operator associated with the form $H^1(\Sigma)\ni\varphi\mapsto \mathcal{Q}^\mathsf{DR}_{\alpha}(\varphi u_{s,\alpha})$ where $u_{s,\alpha}$ is the positive $L^2$-normalized groundstate of the transverse operator $T(s)$. The asymptotic expansion of the effective eigenvalues again follows from perturbation theory and a commutator estimate (see \cite[Section 3]{KKR16} where it is explained how we can estimate such a commutator).
 
Then, we proceed as in the previous section. Note that, by the harmonic approximation, for all $j\geq 1$,
\[\nu_{j}(\alpha)=-\alpha\kappa_{\max}+\alpha^{\frac{1}{2}}\tilde{\nu}_{j}+\mathcal{O}(\alpha^{\frac{1}{4}})\,,\]
where $(\tilde\nu_{j})_{j\in\N^*}$ is the non-decreasing sequence of the eigenvalues of $D^2_{s}+\frac{1}{2}\mathsf{Hess}_{0} (-\kappa)(s,s)$. In particular, the asymptotic gap between consecutive eigenvalues is of order $\alpha^{\frac{1}{2}}$. Then, there exist $c_{0}>0$, $C_{0}>0$ and $C>0$ such that, for $\alpha$ large, $z=\lambda^{\mathsf{DR}}_{j,\eff}(\alpha)+C\in \Zc_{c_{0}, C_{0},\alpha}$. We use Proposition \ref{cor.s-a-Robin} and we get, as in the other examples,
\[|\lambda^{\mathsf{DR}}_{j,\eff}(\alpha)-\lambda^{\mathsf{DR}}_{j}(\alpha)|\leq C\alpha^{-1}\,.\]
\end{proof}

\section{The non-self-adjoint Robin Laplacian between hypersurfaces} \label{sec-Robin}

In this section we prove Theorem \ref{th.n-s-a-Robin}. The proof is split in two main steps. We first transform the problem into an equivalent statement, where $\PP_\e$ is replaced by a unitarily equivalent operator on $\O$.

\subsection{A change of variables}

The operator $\PP_\eps$ is associated to the (coercive) quadratic form defined for $\phi \in H^1(\Omega_{\eps})$ by
\begin{equation}\label{Robin.form}
Q_\eps^1 (\phi) =Q_{\eps,\alpha}^1(\phi)= \int_{\Omega_{\eps}} |\nabla \phi|^2 + \int_{\Sigma_{+,\eps}} \alpha_{+,\eps} |\phi|^2 - \int_{\Sigma_{-,\eps}} \alpha_{-,\eps} |\phi|^2\,.
\end{equation}
As in Section \ref{sec.shrink.Dir}, we use the diffeomorphism $\Theta_\eps$ to see $\PP_\eps$ as an operator on $L^2(\Omega)$: for $\psi \in H^1(\Omega)$ we set 
\[
Q_\eps^2(\psi) = Q_\eps^1(\psi \circ \Theta_\eps^{-1})\,.
\]
We obtain
\begin{align*}
Q_\eps^2(\psi)
& = \int_{\Omega_{\eps}} |(\mathrm{Id}_{T_s\Sigma} + \eps t d_s n)^{-1} \nabla_s \psi (\Theta_\eps^{-1} (x))|^2 \, dx + \frac 1 {\eps^2} \int_{\Omega_{\eps}} |\partial_t \psi (\Theta_\eps^{-1} (x))|^2 \dx x \\
& \quad + \int_{\Sigma_+} \alpha_{+,\eps}  |\psi \circ (\Theta_\eps^+)^{-1}|^2 - \int_{\Sigma_-} \alpha_{-,\eps}  |\psi \circ (\Theta_\eps^+)^{-1}|^2\\
& = \int_{\Omega} \langle G_\eps(s,t) \nabla_s \psi, \nabla_s \psi\rangle_{T\Sigma} \eps\tilde w_{\eps} \dx\sigma \dx t  + \frac 1 \eps \int_{\Omega} |\partial_t \psi|^2\tilde w_{\eps} \dx\sigma \dx t\\
& \quad + \int_{\Sigma} \alpha \, (|\psi|^2\tilde w_{\eps})|_{t=1} \dx\sigma  -  \int_{\Sigma} \alpha  \, ( |\psi|^2\tilde w_{\eps})|_{t=-1} \dx\sigma,
\end{align*}
where, as in \eqref{det-w}, $\tilde w(s,t) = \prod_{j=1}^{d-1} (1- \eps t \kappa_j(s))$. Notice that $L^2(\Omega,\dx\sigma\dx t)$ and $L^2(\Omega, \eps\tilde w_{\eps} \dx\sigma\dx t)$ (or their corresponding Sobolev spaces) are equal as sets, but $\Theta_\eps$ induces only a unitary transformation from $L^2(\Omega, \eps\tilde w_{\eps} \dx\sigma\dx t)$ to $L^2(\Omega_{\eps},\dx x)$.

\subsection{A change of function}

In the next step we make a change of function to turn our problem with Robin boundary conditions into an equivalent problem with Neumann boundary conditions. For this we consider the unitary transform 
\[
\tilde U_{\eps} : \fonc {L^2(\Omega,\eps \tilde w_{\eps} \dx\sigma \dx t)}{L^2(\Omega, e^{-2 \eps t \Re(\alpha)}  \tilde{w}_{\eps} \dx\sigma \dx t),} {u} {\sqrt \eps e^{\alpha \eps t} u.}
\]
We set 
\[
w_{\eps} = e^{-2 \eps t \Re(\alpha)}  \tilde{w}_{\eps}\,.
\]
Then on $H^1(\Omega, w_{\eps} \dx\sigma \dx t)$ we consider the transformed quadratic form given by 
\begin{align*}
Q_\eps (\phi)
& = Q_\eps^2 (\tilde U^{-1} \phi)\\
& = \int_{\Omega} \langle G_\eps (\nabla_s - \eps t \nabla_s \alpha)  \phi, (\nabla_s - \eps t \nabla_s \alpha) \phi\rangle w_{\eps} \dx\sigma \dx t  + \frac 1 {\eps^2} \int_{\Omega} |\partial_t \phi|^2 w_{\eps} \dx\sigma \dx t\\
& \quad - \frac {1}{\eps} \int_{\Omega} \big(\alpha \phi \partial_t \bar \phi + \bar \alpha \bar \phi \partial_t \phi \big) w_{\eps} \dx \sigma \dx t +  \int_{\Omega} |\alpha|^2 |\phi|^2 w_{\eps} \dx\sigma \dx t \\ 
& \quad + \frac 1 \eps \int_{\Sigma} \alpha |\phi| w_{\eps} \dx\sigma - \frac 1 \eps \int_{\Sigma} \alpha |\phi|^2 w_{\eps} \dx\sigma\,.
\end{align*}

By integration by parts we have 
\begin{multline*}
- \frac {1}{\eps} \int_{\Omega} \alpha \phi \partial_t \bar \phi w_{\eps} \dx\sigma \dx t\\
= - \frac 1 \eps \int_{\Sigma} \alpha | \phi| w_{\eps} \dx\sigma + \frac 1 \eps \int_{\Sigma} \alpha |\phi|^2 w_{\eps} \dx\sigma + \int_{\Omega} \left(- 2 \alpha \Re(\alpha) + \frac {\alpha \partial_t \tilde w_{\eps}}{\eps \tilde w_{\eps}} \right) | \phi|^2 w_{\eps} \dx\sigma \dx t\,.
\end{multline*}
Finally,
\begin{align*}
\Qe (\phi)
& = \int_{\Omega} \langle G_\eps (\nabla_s - \eps t \nabla_s \alpha)  \phi, (\nabla_s - \eps t \nabla_s \alpha) \phi\rangle w_{\eps} \dx\sigma \dx t + \frac 1 {\eps^2} \int_{\Omega} |\partial_t \phi|^2 w_{\eps} \dx\sigma \dx t\\
& + \frac {2i}{\eps} \int_{\Omega} \Im(\alpha) \partial_t \phi \bar \phi w_{\eps} \dx\sigma \dx t + \int_{\Omega} V_\eps |\phi|^2 w_{\eps} \dx\sigma \dx t\,,
\end{align*}
where 
\[
V_\eps = |\alpha|^2 - 2 \alpha \Re(\alpha) + \alpha \frac {\partial_t \tilde w_{\eps}}{\eps \tilde w_{\eps}}\,.
\]
On $H^1(\Omega, w_{\eps} \dx\sigma \dx t)$ we can also consider the forms defined by 
\[
\hatQe (\phi) = \int_{\Omega} |\nabla_s \phi|^2 \dx\sigma \dx t  + \frac 1 {\eps^2} \int_{\Omega} |\partial_t \phi|^2 \dx\sigma \dx t  + \int_{\Omega} V_\eff |\phi|^2 \dx \sigma \dx t
\]
and 
\[
\Qeff (\phi) = \int_{\Omega} |\nabla_s \phi|^2 \dx\sigma \dx t + \int_{\Omega} V_\eff |\phi|^2 \dx \sigma \dx t.
\]
We denote by $\LLe$, $\hatLLe$ and $\LLeff$ the operators corresponding to the forms $\Qe$, $\hatQe$ and $\Qeff$, respectively.

\subsection{About the new operator \texorpdfstring{$\LLe$}{L}}
If $U_{\eps}$ denotes the composition of the unitary transform associated with $\Theta_{\eps}$ and $\tilde U_{\eps}$, we write $\LLe =U_{\eps} \PP_\eps U_{\eps}^{-1}$ and the estimate of Theorem \ref{th.n-s-a-Robin} can be rewritten as
\begin{equation} \label{estim-th-He}
\|(\LLe-z)^{-1} - (\LLeff - z)^{-1} \Pi \|_{\mathcal{L}(L^2(\Omega))} \lesssim \eps\,.
\end{equation}
As $\PP_{\eps}$, the operator $\LLe$ is $m$-accretive. We have the following accretivity estimate when $\eps$ goes to $0$.
\begin{lemma} \label{lem-minor-Qe}
If $\eps_0 > 0$ is small enough there exist $M_{0}\geq 0$ and $c_0 > 0$ such that for $\eps \in (0,\eps_0)$, $M \geq M_0$ and $\phi \in H^1(\Omega)$ we have 
\[
\Re \big(\Qe(\phi)\big) + M \|\phi\|_{L^2(\Omega)}^2 \geq c_0 \left( \|\nabla_s \phi\|_{L^2(\Omega)}^2 + \frac 1 {\eps^2} \|\partial_t \phi\|_{L^2(\Omega)}^2 + \|\phi\|_{L^2(\Omega)}^2 \right)\,.
\]
\end{lemma}

\begin{proof}
There exists $C_1 \geq 0$ such that for $\eps \in (0,\eps_0)$ and $\phi \in H^1(\O)$ we have
\begin{align*}
\Re \big(\Qe(\phi)\big)
\geq \ & (1-C_1 \eps) \|(\nabla_s - \eps t \nabla_s \alpha)\phi\|^2 + \frac {1-C_1\eps}{\eps^2} \|\partial_t \phi\|^2\\
& - \frac {2 \|\Im(\alpha)\|_\infty(1+C_1\eps)}\eps  \|\partial_t \phi\|\|\phi\| - C_1 \|\phi\|^2.
\end{align*}
For some $C_2 \geq 0$ we also have 
\[
\|(\nabla_s - \eps t \nabla \alpha) \phi\|^2 \geq (1-\eps) \|\nabla_s \phi\|^2 - C_2 \|\phi\|^2\,.
\]
and
\begin{eqnarray*}
\lefteqn{\frac {1-C_1\e}{\e^2} \nr{\partial_t \f}^2 - \frac {2 \nr{\Im(\a)}_\infty(1+C_1\e)}\e  \nr{\partial_t \f} \nr{\f}}\\
&& = (1-C_1\e) \left( \frac {\nr{\partial_t \f}^2} {\e^2}  - \frac {\nr{\partial_t \f}} {\e}   \frac {2 \nr{\Im(\a)}_\infty(1+C_1\e)\nr{\f}}{1-C_1 \e} \right)\\
&& \geq (1-C_1 \e) \frac {\nr{\partial_t \f}^2} {2 \e^2} - C_2 \nr{\f}^2.
\end{eqnarray*}
The conclusion follows if $\eps_0 > 0$ was chosen small enough.
\end{proof}

A remarkable property of $\LLe$ is the following 
complex symmetry (cf.~\cite{BK}).

\begin{lemma}\label{lem.symmetry}
Let $\e > 0$ and $z \in \C$. If $z \in \C$ is an eigenvalue for $\LLe$ then $\overline{z}$ is an eigenvalue for $\LLe^*$. In particular the operator $\LLe$ has no residual spectrum.
\end{lemma}

\begin{proof}
Since $\LLe$ is unitarily equivalent to $\PP_\eps = \PP_{\e,\a}$, it is sufficient to prove the result for $\PP_{\eps,\a}$. Notice that $\Dom(Q^1_{\eps,\alpha})=\Dom(Q^1_{\eps,\overline{\alpha}})$. Moreover for $\phi,\psi \in \Dom(Q^1_{\e,\a})$ we have $Q^1_{\e,\overline \a}(\phi,\psi) = \overline{Q^1_{\e,\a}(\psi,\phi)}$, so $\PP_{\e,\a}^* = \PP_{\e,\overline \a}$. 
Now let $\psi \in \Dom(\PP_{\e,\a})$. For all $\phi \in \Dom(Q^1_{\e,\a})$ we have 
\[
Q^1_{\e,\overline\a}(\phi,\overline \psi) = \overline{Q^1_{\e,\a}(\overline \phi, \psi)} = \overline{\innp{\overline \phi}{\PP_{\e,\a}\psi}} = \innp{\phi}{\overline{\PP_{\e,\a}\psi}}.
\]
This proves that $\overline \psi \in \Dom(\PP_{\e,\overline\a})$ and $\PP_{\e,\overline \a} \psi = \overline{\PP_{\e,\a} \p}$. Thus, if we denote by $J$ the complex conjugation, 
we get that $\PP_{\e,\a}$ is $J$-self-adjoint 
\[
\PP_{\e,\overline \a} = J \PP_{\e,\a} J.
\]
The conclusion follows.
\end{proof}

\subsection{Proof of Theorem \ref{th.n-s-a-Robin}}
Theorem \ref{th.n-s-a-Robin} will be a consequence of the following proposition.
\begin{proposition}\label{prop.D-eps}
There exist $\eps_{0}, C>0$ such that for all $\eps\in(0,\eps_{0})$, $\varphi\in \mathsf{Dom}(\hatLLe^*)$ and $\psi\in\mathsf{Dom}(\LLe)$,
\[|\Qe (\varphi,\psi)- \hatQe(\varphi,\psi)|\leq C\eps \|\varphi\|_{\hatLLe^*}\|\psi\|_{\LLe}\,.\]
\end{proposition}

\begin{proof}
We set
\[
\mathsf{D}_{\eps}(\varphi,\psi)= \Qe (\varphi,\psi)- \hatQe(\varphi,\psi)\,.
\]
Using the Taylor formula, we get
\[|\mathsf{D}_{\eps}(\varphi,\psi)|\lesssim \eps\|\varphi\|_{H^1_{s}}\|\psi\|_{H^1_{s}}+ \frac 1 \e \|\partial_{t}\varphi\|\|\partial_{t}\psi\|+ \left|\frac {2}{\eps} \int_{\Omega} \Im(\alpha) \partial_t \psi \bar \varphi w_\e \dx\sigma \dx t\right|,
\]
where $\nr{\psi}_{H^1_s}^2 = \nr{\psi}_{L^2(\O)}^2 + \nr{\nabla_s \psi}_{L^2(\O)}^2$.
The most delicate term is the last one.
We have
\[
\abs{\Qe (t\Im(\alpha)\varphi,\psi)- \frac 1 {\eps^{2}} \int\Im(\alpha) \partial_{t}(t\overline{\varphi})\partial_{t}\psi w_\e \dx\sigma \dx t} 
\lesssim\|\varphi\|_{H^1_{s}}\|\psi\|_{H^1_{s}}+\frac 1 \e \|\varphi\|\|\partial_{t}\psi\|,
\]
so
\[
\abs{\Qe (t\Im(\alpha)\varphi,\psi)- \frac 1 {\eps^2} \int\Im(\alpha) \overline{\varphi}\partial_{t}\psi w_\e \dx\sigma \dx t}
\lesssim\|\varphi\|_{H^1_{s}}\|\psi\|_{H^1_{s}}+\frac 1 {\eps^2}\|\partial_{t}\varphi\|\|\partial_{t}\psi\|+\frac 1 {\eps}\|\varphi\|\|\partial_{t}\psi\|\,.
\]
Since $\Qe (t\Im(\alpha)\varphi,\psi) = \innp{t \Im(\a) \varphi}{\LLe \p}$, we obtain 
\[
\left|\frac 1 \e \int \Im(\alpha)\partial_{t}\psi\overline{\varphi} w_\e \dx\sigma \dx t\right|\lesssim \eps\|\varphi\|_{H^1_{s}}\|\psi\|_{H^1_{s}}+ \frac 1 \e \|\partial_{t}\varphi\|\|\partial_{t}\psi\|+\|\varphi\|\|\partial_{t}\psi\|+\eps\|\LLe \p \|\|\varphi\|\,.
\]
We conclude with Lemma \ref{lem-minor-Qe}.
\end{proof}

By Proposition \ref{prop.D-eps}, there exists $C \geq 0$ such that for $z \in K$, $\varphi\in \mathsf{Dom}((H^\eff_{\eps})^*)$ and $\psi\in\mathsf{Dom}(H_{\eps})$ we have 
\[|Q^\eff_{\eps}(\varphi,\psi)-z\langle\varphi,\psi\rangle-\left(Q_{\eps}(\varphi,\psi)-z\langle\varphi,\psi\rangle\right)|\leq C \eps \|\varphi\|_{(H_{\eps}^\eff-z)^*}\|\psi\|_{H_{\eps}-z}\,.\]
Finally, we apply Lemma \ref{lem-res-bot} and Lemma \ref{lem-diff-resolv}, and Theorem \ref{th.n-s-a-Robin} follows.

\section*{Acknowledgment}
The authors would like to thank 
the \emph{Centre International de Rencontres Math\'ematiques} in Marseille
for the support during two weeks in August and September 2016
where this research was initiated in the framework of 
a Research-In-Pairs grant. The research of P.S.\ was supported by the \emph{Swiss National Science Foundation}, 
SNF Ambizione grant No. PZ00P2\_154786. D.K. and P.S. were partially supported by FCT (Portugal)
through project PTDC/MAT-CAL/4334/2014.

\bibliographystyle{acm}
\bibliography{v15}

\end{document}